\documentclass[12pt]{article}
\usepackage{apacite,latexsym,graphics,epsfig,amsmath,amssymb,tabularx,booktabs, color, mathtools,tikz}
\usepackage{pgf}

\usepackage{natbib}
\usepackage{enumerate}
\usepackage{amsthm}
\usepackage{url}
\usepackage{fancyhdr,titling}
\usepackage[mathscr]{eucal}
\usetikzlibrary{trees}

\tikzstyle{bag} = [align=center]

\usepackage[T1]{fontenc}
\theoremstyle{definition}

\theoremstyle{remark}
\newtheorem{example}{Example}[section]
\theoremstyle{plain}
\newtheorem{theorem}{Theorem}[section]
\newtheorem{lemma}[theorem]{Lemma}

\newtheorem{corollary}[theorem]{Corollary}

\newcommand{\tr}{^{\prime}}

\makeatletter
\newcommand{\oset}[2]{%
  {\mathop{#2}\limits^{\vbox to -.1\ex@{\kern-\tw@\ex@
   \hbox{\scriptsize #1}\vss}}}}
   
\def\keywords#1{{\vskip4pt
\noindent
\hbox to50.5pt{KEYWORDS:\quad\hss}\vtop{\advance \hsize by -59.5pt
\leftskip=28pt \rightskip=0pt
\noindent\ignorespaces#1\vskip8pt}}}

\newcommand*\xbar[1]{%
  \hbox{%
    \vbox{%
      \hrule height 0.5pt 
      \kern0.5ex
      \hbox{%
        \kern-0.25em
        \ensuremath{#1}%
        \kern-0.1em
      }%
    }%
  }%
}

\makeatletter
\let\runauthor\@author
\let\runtitle 

\date{}

\begin{document}

\title{On the maximum likelihood estimation in general log-linear models}

\author{Anna Klimova \\
{\small{National Center for Tumor Diseases (NCT), Partner Site Dresden, and}}\\
{\small{Institute for  Medical Informatics and Biometry,}}\\ 
{\small{Technical University, Dresden, Germany} }\\
{\small \texttt{anna.klimova@nct-dresden.de}}\\
{}\\
\and
Matthias Kuhn \\
{\small{Institute for  Medical Informatics and Biometry,}}\\ 
{\small{Technical University, Dresden, Germany} }\\
{\small \texttt{matthias.kuhn@tu-dresden.de}}\\
{}
}

\maketitle

\begin{abstract}
General log-linear models specified by non-negative integer design matrices have a potentially wide range of applications, although using  models without the genuine overall effect,  that is,  ones which cannot be reparameterized to include a normalizing constant,  is still rare.  The log-linear models without the overall effect arise naturally in practice,  and can be handled in a similar manner to models with the overall effect.  A novel iterative scaling procedure for the MLE computation under such models is proposed, and its convergence is proved.  The results are illustrated using data from a recent clinical study.
\end{abstract}


\baselineskip=18pt

\section{Introduction}

The general log-linear models for contingency tables are generated by design matrices with non-negative integer entries and are a generalization of the hierarchical log-linear models \citep[cf.][]{Haber}.  Let $\mathcal{I}$ be  a complete or incomplete contingency table, and $\mathbf{A}$ be  a design matrix. A general log-linear model is the set of distributions on $\mathcal{I}$ which satisfy: 
\begin{equation}\label{mainModel}
\log \boldsymbol{\delta} = \mathbf{A}\tr \boldsymbol{\beta},
\end{equation}
where $\boldsymbol{\delta}$ is a vector of parameters  associated with cells in $\mathcal{I}$, and $\boldsymbol \beta$ are log-linear parameters. Among special cases of (\ref{mainModel}) are conventional log-linear models \citep*{BFH} and relational models \citep*{KRD11}.  In the first case, $\mathbf{A}$  comprises the indicators of cylinder sets of marginal distributions of $\mathcal{I}$, and the components of $\boldsymbol \beta$ are associated to these sets. In the second,  the rows of $\mathbf{A}$ are indicators of arbitrary subsets of $\mathcal{I}$, and $\boldsymbol \beta$ are the subset parameters. 

The structure of a hierarchical log-linear model presumes that the row space of the design matrix contains a vector of $1$'s, and, therefore, implies the existence of a genuine common effect present in all cells. This effect is conventionally included in many statistical models in the form of intercept, and plays the role of a normalizing constant,  to ensure that the probability distributions sum to $1$.  
The presence of the overall effect is not compulsory though in the framework of relational models proposed in   \cite{KRD11}. In certain cases, the overall effect assumption may conflict with the model structure, by partially destroying it after the overall effect is added to the model \citep{KRoveff}. Because the overall effect assumption, explicit or hidden, is almost automatic, no systematic treatment of general log-linear models without the overall effect,  with the exception of relational models,  has been given so far. A more careful look at their properties is necessary, because, as already was established for relational models, the properties of models with and without the overall effect are intrinsically different, far beyond the proofs being more laborious. 

The MLE computation under  general log-linear models can, of course,  be performed using the Newton-Raphson's or algorithms  for convex optimization \citep[cf.][]{BertsekasNLP, AitchSilvey60, EvansForcina11}. Although these algorithms are usually fast, they become unstable on the boundary of a parameter space \citep{Agresti2002}. Another, more stable, option is to employ  a log-linear model-class specific algorithms, such as the iterative proportional fitting (IPF) or its generalizations, although some of them are only applicable where the design matrix consists of only 0-1 entries and the overall-effect assumption is met \citep*[cf.][]{BFH, Agresti2002, Huang2010}. Generalized Iterative Scaling (GIS) of \cite{DarrochRatcliff}  was originally proposed for the MLE computation under log-linear models with non-negative, but not necessarily 0-1, model matrices, although the presence of the overall effect was critical in the proof of convergence. Improved Iterative Scaling (IIS) was developed by \cite*{DDL1997} as a generalization of GIS in the context of feature selection. Although, formally, IIS did not rely on the assumption of constant sum of features, it does not converge to the MLE in the no-overall effect case \citep{KRipf1}.  Generalized iterative proportional fitting (G-IPF) of \cite{KRipf1} was specifically designed to handle relational models without the overall effect,  but its convergence relies on the model matrix having exclusively 0-1 entries.  In this paper, we aim to develop an algorithm which would combine the merits of GIS and G-IPF and can be used to compute the MLEs for general log-linear models without the overall effect. 

In Section \ref{sectionNotation}, the relevant background on the ML estimation is summarized, with the emphasis on the differences between the models with and without the genuine overall effect.  Section \ref{MLEcompute} turns to the MLE computation via the iterative proportional fitting approach and presents the main contribution of this paper: a proof of convergence of the generalized  iterative scaling of \cite{DarrochRatcliff} for general log-linear models without the overall effect, either for probability or intensity distributions. A two-step procedure for fitting the MLE is then described and illustrated using a number of examples.   Finally,  the framework is applied to the data collected in a clinical study.

\section{Background and notation}\label{sectionNotation}

Let $Y_1, \dots, Y_K$ be  random variables taking values in finite sets $\mathcal{Y}_1, \dots, \mathcal{Y}_K$, respectively. A non-empty set $\mathcal{I}  \subseteq \mathcal{Y}_1 \times  \dots \times \mathcal{Y}_K$  will be referred to as a table, and each element of $\mathcal{I}$ will be called a cell. The order of cells in $\mathcal{I}$ is fixed, and each cell is indexed with $i$, for $i \in \mathcal{I}$.   Assume that the population distribution is parameterized by a strictly positive $\boldsymbol \delta =(\delta_i)_{i = 1}^I$, where $I = |\mathcal{I}|$ is the total number of cells.   In this paper, $\boldsymbol \delta$ is either a probability vector, that is, $\delta_i \equiv p_i \in (0,1)$, with $\sum_{i =1}^I p_i = 1$, or a vector of intensities, $\delta_i \equiv \lambda_i > 0$, for all $i \in \mathcal{I}$. The set of all strictly positive distributions on $\mathcal{I}$ will be denoted by $\mathcal{P}$. 

Let $\mathbf{A} = (a_{ji}) \in \mathbb{Z}_{\geq 0}^{J \times I}$ be a $J \times I$ matrix of rank $J$ with no zero columns. The general log-linear model generated by $\mathbf{A}$, denoted $\mathcal{G}(\mathbf{A})$, is the subset of $\mathcal{P}$ which satisfies:
\begin{equation} \label{PMmatr}
\mathcal{G}(\mathbf{A}) = \left \{ \boldsymbol \delta \in \mathcal{P}: \,\, \mbox{log } \boldsymbol \delta = \mathbf{A}'\boldsymbol \beta, \, \mbox{for some } \, \boldsymbol \beta \in \mathbb{R}^J \right\}.
\end{equation}
Here, the components of $\boldsymbol \beta$ are the log-linear parameters of the model. In the conventional set-up, $\mathcal{G}(\mathbf{A})$ has a parameter, called the overall effect, or the normalization constant, which is common to all cells in $\mathcal{I}$. In such case, the model can be reparametrized to have $\boldsymbol 1\tr$ in the row span of the  model matrix. This manuscript will mainly focus on models without the overall effect, that is, when $\boldsymbol 1\tr \notin rowspan (\mathbf{A})$.  A dual representation of  log-linear model (\ref{PMmatr}) can be obtained from an integer matrix, $\mathbf{D}$, whose rows are a basis of $Ker(\mathbf{A})$:
\begin{equation}\label{dualF2d}
\mathcal{G}(\mathbf{A}) = \{ {\boldsymbol \delta} \in \mathcal{P}: \,\, \mathbf{D} \mbox{log }\boldsymbol \delta = \boldsymbol 0\}.
\end{equation}

A general log-linear model, whether with or without the overall effect, is also an exponential (multiplicative) family of distributions:
\begin{equation} \label{RMexpF}
\mathcal{G}(\mathbf{A}) = \{ \boldsymbol \delta \in \mathcal{P}: \,\,\delta_i = \exp\{\sum_{j=1}^J  a_{ji} \beta_j\} = \prod_{j=1}^J \theta_j^{a_{ji}}, \, i \in \mathcal{I}, \, \mbox{ for some } \, \boldsymbol \theta \in \mathbb{R}^J_{>0}\},
\end{equation}
where $\theta_j = \mbox{exp  }(\beta_j)$, and $\boldsymbol \theta =(\theta_1, \dots, \theta_J) \in \mathbb{R}_{>0}^J$ denotes the vector of (multiplicative) parameters associated to columns of $\mathbf{A}$.  In the case of intensities, whether or not $\boldsymbol 1\tr \in rowspan (\mathbf{A})$,  (\ref{RMexpF}) is a regular exponential family of order $J$. In the case of probabilities, the additional normalization constraint, $\boldsymbol 1\tr \boldsymbol \delta = 1$ is imposed, so  (\ref{RMexpF}) is an exponential family of order $J-1$.  When $\boldsymbol 1\tr \in rowspan (\mathbf{A})$, the family for probabilities is regular, and, otherwise, it is curved \citep[see e.g.][p.229]{Rudin,KassVos}. As shown in the next section, the properties of maximum likelihood estimators under $\mathcal{G}(\mathbf{A})$ depend very strongly of whether the model is a regular or curved exponential family.


Let $\mathbf{Y} = (Y_1, \dots, Y_K)$ be a random variable that has either a multivariate Poisson distribution $Pois(\boldsymbol \lambda)$ or a multinomial distribution $Mult(N, \boldsymbol p)$, and $\boldsymbol y$ be a realization of $\mathbf{Y}$. In the Poisson case, set $\boldsymbol q = 
\boldsymbol y$, and, in the multinomial case, denote $\boldsymbol q = 
\boldsymbol y/ N$.  Assume that the MLE of the cell parameters  under the model $\mathcal{G}(\mathbf{A})$ exists. Then, the MLE $\hat{\boldsymbol \delta}_{\boldsymbol q}$ is the unique solution to the system of equations \citep[cf.][]{KRD11}:
\begin{align}\label{MLEsys}
&\mathbf{A}\boldsymbol{\delta} = \gamma \mathbf{A} \boldsymbol q, \quad \mathbf{D} \mbox{log } \boldsymbol{\delta} = \boldsymbol 0, \quad \mbox{and only for } \boldsymbol \delta \equiv \boldsymbol p: \,\,\boldsymbol 1'\boldsymbol \delta = 1. 
\end{align}
If $\mathcal{G}(\mathbf{A})$ is a regular exponential family, the Birch theorem holds, namely, the sufficient   statistics,  $\mathbf{A} {\boldsymbol q}$, are preserved by the MLE, implying $\gamma \equiv 1$  \citep[cf.][]{BirchMLE, Andersen74, Haberman}.    When $\mathcal{G}(\mathbf{A})$ is a curved exponential family (in our setting, it would be a model for probabilities without the overall effect), $\gamma = \gamma(\boldsymbol q)$,  depends on the data. In this case, a more general version of Birch theorem holds, and $\gamma$ is referred to as the adjustment factor of the MLE \citep{KRD11, Forcina2019}.  An illustration is given next.

\begin{example} \label{ExampleVaccine}
Consider the model under which the cell probabilities $\boldsymbol  p = (p_1, p_2, p_3, p_4)\tr$ are parameterized according to the tree in Figure \ref{TreeVaccine}, where $\theta_0, \theta_1 \in (0,1)$, $\theta_1= 1-\theta_0$:
$$p_{1} = \theta_0^3, \,\,  p_{2} = \theta_0^2 \theta_1, \,\, p_{3} = \theta_0\theta_1, \,\, p_{4} = \theta_1.$$
Equivalently, the model is can be expressed in the log-linear form (\ref{PMmatr}), with the model matrx
\begin{equation*}
\mathbf{A} = \left( 
\begin{array}{ccccc}
3&2&1&0\\
0&1&1&1\\
\end{array}
\right).
\end{equation*}
Because $\boldsymbol 1 \tr \notin rowspan(\mathbf{A})$, the model does not have the genuine overall effect. 
Let $\boldsymbol y = (y_1, y_2, y_3, y_4)\tr$ be a realization of $Mult(N, \boldsymbol p)$, where  $N = \sum_{i = 1}^4 y_i$. The kernel of the log-likelihood function is equal to
$L(\boldsymbol \theta \mid \boldsymbol y) = (3y_{1}+2y_{2} +  y_{3})\mbox{log } \theta_0 +  (y_{2}+y_{3} +  y_{4}) \mbox{log } \theta_1$. One can show that it 
is maximized by   $\hat{\theta}_0 = z_1/ z_3$, $\hat{\theta}_1 = z_2/ z_3$, where  
$z_1 = 3y_1 + 2 y_2 + y_3$, $z_2 =     y_2+y_3 + y_4$, and $z_3 = z_1 + z_2 = 3y_1 + 3 y_2 + 2 y_3 + y_4$. Accordingly, the MLE $\hat{\boldsymbol p}$ is equal to
$\hat{\boldsymbol p} = \left ({z_1^3}/{z_3^3}, {z_1^2z_2}/{z_3^3}, {z_1 z_2}/{z_3^2},  {z_2}/{z_3}\right)\tr.$
One can check that $\mathbf{A} \hat{\boldsymbol p} = \gamma_{y} \mathbf{A} (\boldsymbol y/N)$, so the adjustment factor $\gamma_y =  N \cdot (z_1^2 + z_1z_3 + z_3^2)/{z_3^3}.$  
\qed
\end{example}

\begin{figure}
\begin{center}
\includegraphics[scale=0.6]{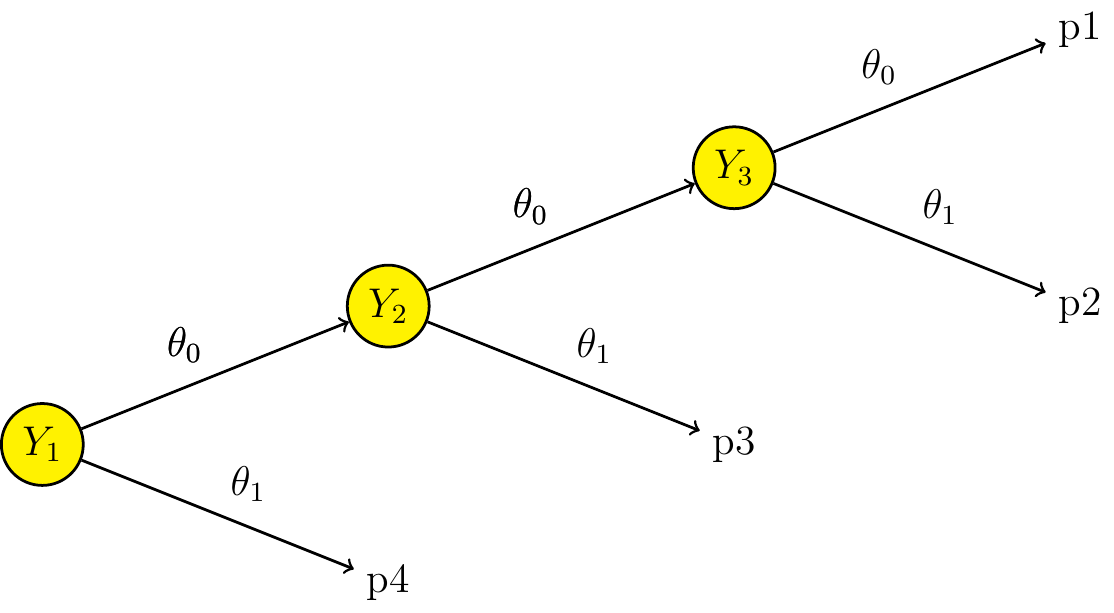}
\end{center}
\caption{Tree representation for the model in Example \ref{ExampleVaccine}.}
\label{TreeVaccine}
\end{figure}

Log-linear models for which no closed-form MLE exists occur  in practice very often. For such models, the MLE is approximated using a numeric procedure. In the next section, an iterative scaling algorithm for MLE computation is presented and its convergence is proved. 

\section{Generalized Iterative Scaling for the MLE computation}\label{MLEcompute}

Let the data $\boldsymbol y$ come from either $Pois(\boldsymbol \lambda)$ or  $Mult(N, \boldsymbol p)$, and $\mathcal{G}(\mathbf{A})$ be a general log-linear model. Assume that the MLE, $\hat{\boldsymbol \lambda}$ or $\hat{\boldsymbol p}$, respectively, under $\mathcal{G}(\mathbf{A})$, given $\boldsymbol y$, exists. In this paper, we take the log-linear model-specific approach to the MLE calculation and focus on the generalized iterative scaling (GIS) algorithm proposed by \cite{DarrochRatcliff}

Let $\mathbf{A} = (a_{ji})$ be a non-negative real $J \times I$ matrix such that the sum of its row vectors is the vector of all $1$'s:
\begin{equation}\label{star}
 (\sum_{j = 1}^J a_{j1},\dots, \sum_{j = 1}^J a_{jI}) = \boldsymbol 1.
\end{equation}
Equivalently,  all column sums of $\mathbf{A}$ are equal to $1$. For a given $\boldsymbol k \in \mathbb{R}^J_{> 0}$ and $\boldsymbol \pi \geq \boldsymbol 0$, $\boldsymbol 1\tr \boldsymbol \pi \leq 1$, GIS aims to find a probability distribution $\boldsymbol p \in \mathcal{P}$ which satisfies
\begin{equation} \label{DRsetup}
p_i = \pi_i \mu \prod_{j=1}^J   \theta_j^{a_{ji}}, \,\, \mathbf{A} \boldsymbol p = \boldsymbol k. 
\end{equation}
By Lemma 4 in  \citep{DarrochRatcliff}, the constraints in (\ref{DRsetup}) and can be transformed into
\begin{align} \label{DRsetup2}
& p_i = \pi_i \prod_{j=1}^J   \eta_j^{a_{ji}}, \,\,  \mathbf{A} \boldsymbol p = \boldsymbol h, 
\end{align}
where $\sum_{j = 1}^J a_{ji}= 1$, for each $i \in 1, \dots, I$, and  $\boldsymbol 1\tr \boldsymbol h =1$.
In their Theorem 1, \citep{DarrochRatcliff} showed that the solution to (\ref{DRsetup2}), if exists, is unique and can be found as the limit, as $n \to \infty$, of the sequence
\begin{align} \label{DRIter} 
&p^{(n+1)}_i = p^{(n)}_i\prod_{j =1}^J \left( \frac{ h_j}{ h_j^{(n)}}\right)^{a_{ji}}, \hspace{2mm}  \mbox{ for all } i \in \mathcal{I}, 
\end{align} 
where $p^{(0)}_i = \pi_i$ and $h_j^{(n)} = A_j \boldsymbol p^{(n)}$. 

The proof of convergence is built upon the properties of the Kullback-Leibler information divergence between two probability distributions, namely,  
\begin{equation} \label{KLdef} 
\mathcal{K}(\boldsymbol p, \boldsymbol q) = \sum_{i \in \mathcal{I}} p_i\mbox{log }(p_i/q_i),
\end{equation}
where $\boldsymbol p, \boldsymbol q \in \bar{\mathcal{P}}$; by convention, $0/0=1$ and $0\log 0 = 1$. Assuming that the constraints  (\ref{DRsetup2}) are consistent, \cite{DarrochRatcliff} first show that for any distribution $\boldsymbol q$ satisfying (\ref{DRsetup2}), 
$$\mathcal{K}(\boldsymbol q, \boldsymbol p^{(n+1)}) = \mathcal{K}(\boldsymbol q, \boldsymbol p^{(n)}) - \mathcal{K}(\boldsymbol h, \boldsymbol h^{(n)}),$$
and, therefore, because $\mathcal{K}(\boldsymbol p, \boldsymbol q) \geq 0$,  the sequence  $\mathcal{K}(\boldsymbol q, \boldsymbol p^{(n)})$ is monotone decreasing, as $n \to \infty$, and $\mathcal{K}(\boldsymbol h, \boldsymbol h^{(n)}) \to 0$, finally concluding that $\boldsymbol h^{(n)} \to \boldsymbol h$. See also \cite{Csiszar}. 

Two points in this proof deserve attention. Firstly, the condition (\ref{star}) is equivalent to $\mathbf{A}$ having the overall effect, and, secondly, when this condition is not met, it may happen that $\boldsymbol 1\tr \boldsymbol p^{(n)} > 1$ and $\mathcal{K}(\boldsymbol q, \boldsymbol p^{(n)}) < 0$. It is illustrated  in forthcoming Figure \ref{Ex41prob} of this paper. These issues were taken into account  in \cite{KRipf1}, who developed an iterative scaling procedure, IPF($\gamma$),  as an instantiation of the algorithm proposed by \cite{Bregman} to find a common point of convex sets \citep{KRipf1}, and, consequently, proved its convergence with respect to  the Bregman divergence associated with the function $F(\boldsymbol x) = \sum_{i \in \mathcal{I}} x_i \mbox{log }x_i$. Let $\boldsymbol t, \boldsymbol u \in \mathbb{R}^{|\mathcal{I}|}_{>0}\,$ and $\mathcal{D}(\boldsymbol t, \boldsymbol u)$ denote 
\begin{equation} \label{BDdef} 
\mathcal{D}(\boldsymbol t, \boldsymbol u) = \sum_{i \in \mathcal{I}} t_i \mbox{log }(t_i/u_i) - (\sum_{i \in \mathcal{I}} t_i - \sum_{i \in \mathcal{I}} u_i).
\end{equation}
A similar approach will be implemented in this paper for the general log-linear model (\ref{PMmatr}). It will be shown next that the procedure of Darroch and Ratcliff, after a certain modification, can indeed be used for the MLE computation under in the no-overall-effect case, not only for probabilities and but also for intensities.  

Let $\mathbf{A}$ be a non-negative integer matrix of full row rank, and $\mathbf{A}_1 = \mathbf{A}/\| \mathbf{A} \|_1$, where $\| \mathbf{A} \|_1$ stands for the $L_1$ norm of $\mathbf{A}$
(its maximal column sum):
 $\| \mathbf{A} \|_1 = \underset{1 \leq j \leq I}{{max}} \sum_{j = 1}^J |a_{ji}|.$
Let  $\boldsymbol q \in \mathbb{R}^I_{\geq 0}$, such that $A_1 \boldsymbol q, \dots, A_J \boldsymbol q > 0$, where $A_1, \dots, A_J$ denote the rows of $\mathbf{A}$. For the sake of brevity, it is assumed that $\mathbf{A}$ coincides with with its normalized version:    $ \mathbf{A} \equiv \mathbf{A}/\| \mathbf{A} \|_1$. 

For a fixed $\gamma > 0$, consider a sequence of vectors $\boldsymbol \delta^{(n)}_{\gamma}$ obtained during the following procedure, referred to in the sequel as GIS($\gamma$) Algorithm. Set ${\delta}_{\gamma,i}^{(0)} = 1$ for all $i \in \mathcal{I}$. Then iterate
\begin{eqnarray} 
\delta_i^{(n+1)} &=& \delta_i^{(n)} \prod_{j = 1}^J \left[\gamma\frac{{A}_{j} \boldsymbol{q}}{{A}_{j} \boldsymbol{\delta}^{(n)}}\right]^{a_{ji}}
  \,\, \mbox{for all } i \in \mathcal{I};   \label{RipfGamma} 
\end{eqnarray}
until  $\gamma A_{j}\boldsymbol{q} = A_{j}\boldsymbol{\delta}_{\gamma}^{(n+1)}$ for all $j$,  set $n = n+1$ up to a desired precision. 
After convergence, set $\boldsymbol{\delta}_{\gamma}^{*}=\boldsymbol{\delta}_{\gamma}^{(n)}$, and finish. \qed

\begin{theorem}\label{ThGammaNew}
The sequence  $\boldsymbol{\delta}_{\gamma}^{(n)}$, obtained from GIS($\gamma$), converges, as $n \to \infty$, and 
the limit  $\boldsymbol{\delta}_{\gamma}^{*}$ satisfies: (i) $\mathbf{A}\boldsymbol{\delta}_{\gamma}^* = \gamma \mathbf{A} \boldsymbol q$, (ii) $\mathbf{D} \mbox{log } \boldsymbol{\delta}_{\gamma}^* = \boldsymbol 0$.
\end{theorem}

 The proof of the theorem will be constructed from a series of lemmas, presented next.
 
 \begin{lemma}\label{LemmaDeltaInequality}
 Let $\mathbf{A}$ be the matrix defined in Theorem \ref{ThGammaNew}, $\boldsymbol{\delta}_{\gamma}^{(n)}$ is the sequence in (\ref{RipfGamma}), with $\gamma = 1$. Then,
  $$\sum_{i = 1}^I \delta^{(n+1)}_i \leq \sum_{i = 1}^I \delta^{(n)}_i +\sum_{j = 1}^J (A_j \boldsymbol q - A_j\boldsymbol \delta^{(n)}).$$
  \end{lemma}

\begin{proof}
  
In the proof, the weighted arithmetic mean - geometric mean inequality:
\begin{equation}\label{wAGinequality}
 \frac{w_1 x_1 + \dots w_n x_n}{w} \geq \left (x_1^{w_1} \cdots x_n^{w_n}  \right)^{1/w},
 \end{equation}
 will be used.  Here $x_1, \dots x_n > 0$, and $w_i \geq 0$, such that $w_1 + \cdots + w_n = w$.

Notice that $$\delta_i^{(n+1)} = \delta_i^{(n)} \prod_{j = 1}^J \left[\frac{{A}_{j} \boldsymbol{q}}{{A}_{j} \boldsymbol{\delta}^{(n)}}\right]^{a_{ji}} \cdot 1^{1-\sum_{j=1}^J a_{ij}},$$
and apply the inequality (\ref{wAGinequality}):
$$\delta_i^{(n+1)} \leq \delta_i^{(n)} \left( \sum_{j = 1}^J a_{ji} \frac{{A}_{j} \boldsymbol{q}}{{A}_{j} \boldsymbol{\delta}^{(n)}}  + {1-\sum_{j=1}^J a_{ji}}\right) = \sum_{j = 1}^J a_{ji}\delta_i^{(n)}  \frac{{A}_{j} \boldsymbol{q}}{{A}_{j} \boldsymbol{\delta}^{(n)}}  + \delta_i^{(n)} -\sum_{j=1}^J a_{ji}\delta_i^{(n)}.$$   
Therefore,
\begin{align*}
\sum_{i = 1}^I\delta_i^{(n+1)} &\leq \sum_{i = 1}^I \sum_{j = 1}^J a_{ji}\delta_i^{(n)}  \frac{{A}_{j} \boldsymbol{q}}{{A}_{j} \boldsymbol{\delta}^{(n)}}  +  \sum_{i =1}^I\delta_i^{(n)} -  \sum_{i =1}^I\sum_{j=1}^J a_{ji}\delta_i^{(n)}\\
& = \sum_{i =1}^I\delta_i^{(n)} +  \sum_{j = 1}^J   \frac{{A}_{j} \boldsymbol{q}}{{A}_{j} \boldsymbol{\delta}^{(n)}} A_j \boldsymbol \delta^{(n)} - \sum_{j = 1}^J A_j \boldsymbol \delta^{(n)}  = \sum_{i =1}^I\delta_i^{(n)} + \sum_{j = 1}^J (A_j \boldsymbol q - A_j\boldsymbol \delta^{(n)}).
\end{align*}                                           
 \end{proof}
 
 In the next lemma, whose proof is deferred to the Appendix, it will be shown that the sequence $\boldsymbol \delta^{(n)}$, in fact, converges to a projection of $\boldsymbol q$ on the cone $\mathbf{A} \boldsymbol \delta = \mathbf{A} {\boldsymbol q}$. 
 \begin{lemma}\label{LemmaBregmanInequality}
 Let $\mathbf{A}$ be the matrix defined in Theorem \ref{ThGammaNew}, $\boldsymbol{\delta}_{\gamma}^{(n)}$ is the sequence in (\ref{RipfGamma}), with $\gamma = 1$. Then, for any $\boldsymbol z \in \mathbb{R}^I_{>0}$, \quad
  $\mathcal{D}(\boldsymbol z, \boldsymbol \delta^{(n+1)}) \leq  \mathcal{D}(\boldsymbol z, \boldsymbol \delta^{(n)}) - \mathcal{D}(\mathbf{A}\boldsymbol q, \mathbf{A}\boldsymbol \delta^{(n)}).$
  \end{lemma}

 A proof of Theorem \ref{ThGammaNew} is given now. 
  
\begin{proof}
  
Because the sequence $\mathcal{D}(\boldsymbol q, \boldsymbol \delta^{(n)})$ is bounded by zero from below and, by Lemma \ref{LemmaBregmanInequality}, monotone decreasing,   $\mathcal{D}(\mathbf{A}\boldsymbol q, \mathbf{A}\boldsymbol \delta^{(n)}) \to 0$ as $n \to \infty$, which implies that 
 $\mathbf{A}\boldsymbol \delta^{(n)} \to \mathbf{A}\boldsymbol q$.

  
(i) In more generality, when $0 < \gamma \neq 1$,  after replacing $\boldsymbol q$ with $\gamma \boldsymbol q$, and, respectively, $\boldsymbol{\delta}^{(d)}$ with $\boldsymbol{\delta}_{\gamma}^{(d)}$, one has 
   $$\mathcal{D}(\boldsymbol z, \boldsymbol \delta^{(n+1)}) \leq \mathcal{D}(\boldsymbol z, \boldsymbol \delta^{(n)}) - 
 \sum_{j = 1}^J  A_j \boldsymbol z \cdot \log \left[\frac{\gamma{A}_{j} \boldsymbol{q}}{{A}_{j} \boldsymbol{\delta}^{(n)}}\right] + \sum_{j = 1}^J  \gamma {A}_{j} \boldsymbol{q} - \sum_{j = 1}^J A_j \boldsymbol \delta^{(n)},$$
 and for a $\boldsymbol z \in \mathbb{R}^I_{>0}$, that $\mathbf{A} \boldsymbol z = \gamma \mathbf{A} \boldsymbol q$:
 $$\mathcal{D}(\boldsymbol z, \boldsymbol \delta^{(n+1)}) \leq  \mathcal{D}(\boldsymbol q, \boldsymbol \delta^{(n)}) - \mathcal{D}(\gamma \mathbf{A}\boldsymbol q, \mathbf{A}\boldsymbol \delta^{(n)}).$$
Analogous to the above, $\mathbf{A}\boldsymbol \delta^{(n)} \to \gamma \mathbf{A}\boldsymbol q$, so one concludes that the sequence $\boldsymbol{\delta}_{\gamma}^{(d)}$ converges and its limit, $\boldsymbol{\delta}_{\gamma}^*$, satisfies 
$\mathbf{A}\boldsymbol{\delta}_{\gamma}^* = \gamma\boldsymbol A \boldsymbol{q}.$ 
(ii) It can be proven by induction that during the coordinate transformation, the values of the generalized odds ratios do not change. Since ${\delta}_{\gamma,i}^{(0)} = 1$, for all $i \in \mathcal{I}$, $\mathbf{D} \mbox{log } \boldsymbol {\delta}_{\gamma}^{(0)} = \boldsymbol 0$,  the statement holds for $d = 0$. Assume that $\mathbf{D} \mbox{log } \boldsymbol{\delta}_{\gamma}^{(d)} = \boldsymbol 0$ for a positive integer $d$. Set $C_j = \frac{\gamma {A}_{j} \boldsymbol{q}}{{A}_{j} \boldsymbol{\delta}_{\gamma}^{(d)}}$. Then, 
\begin{eqnarray*}
\mathbf{D} \mbox{log } \boldsymbol{\delta}_{\gamma}^{(d+1)} &=& 
 \mathbf{D} \mbox{log } \boldsymbol{\delta}_{\gamma}^{(d)} + \sum_{j = 1}^J\mbox{log } C_j \mathbf{D}A'_j = \boldsymbol 0,
\end{eqnarray*}
because $\mathbf{D}$ is a kernel basis matrix and thus $\mathbf{D}A'_{j}= \boldsymbol 0$.
Therefore, $\mathbf{D} \mbox{log } \boldsymbol{\delta}_{\gamma}^{(d)} = \boldsymbol 0$ for all $d = 0, 1, 2 \dots$, and, by continuity of matrix multiplication and logarithm, $\mathbf{D} \mbox{log } \boldsymbol{\delta}_{\gamma}^{*} = \boldsymbol 0$. 
 \end{proof}

\begin{corollary}
Let $\mathcal{G}(\mathbf{B})$ be a log-linear model for intensities, $\boldsymbol y$ be a realization of $\boldsymbol Y \sim Pois(\boldsymbol \lambda)$, and assume that the MLE $\hat{\boldsymbol \lambda}$ exists. Then, the sequence $\boldsymbol \delta^{(n)}_1$ obtained from $GIS(1)$ with $\mathbf{A} = \mathbf{B}/\|\mathbf{B} \|_1$ and $\boldsymbol q \equiv \boldsymbol y$ converges  to the MLE   $\hat{\boldsymbol \lambda}$ under the model $\mathcal{G}(\mathbf{B})$.
\end{corollary}
\noindent Because the MLE $\hat{\boldsymbol \lambda}$  is the unique solution to (\ref{MLEsys}), the statement immediately follows.

Clearly, GIS($\gamma$) does not guarantee that the normalization condition holds, that is, the limit vector does not necessarily sum to 1, and, therefore, is not sufficient for models for probabilities without the overall effect. However, the algorithm can be employed as a ``core step'' of the MLE computation, similarly to how it is performed for relational models for probabilities without the overall effect \citep{KRipf1}. The overall idea is to  construct an algorithm with two steps: Core and Adjustment. A core step, consisting of running GIS($\gamma$)  for a fixed $\gamma$, is alternated with an ``adjustment step'', during which the current value of $\gamma$ is updated, to ensure that the resulting total converges to $1$. For example, the adjustment step can be implemented using the Newton's step of  \cite{Forcina2019}. The algorithm is terminated after the total is close to $1$ up to a desired prescision. The MLE computation via the two-step procedure is summarized in the next theorem:

\begin{theorem}\label{G_IPFepsConv}
 Let $\mathcal{G}(\mathbf{B})$ be a log-linear model for probabilities, $\boldsymbol y$ be a realization of $Mult(N, \boldsymbol p)$, and assume that the MLE $\hat{\boldsymbol p}$, given $\boldsymbol y$, exists. Take $\mathbf{A} = \mathbf{B}/\|\mathbf{B} \|_1$ and $\boldsymbol q \equiv \boldsymbol y/(\boldsymbol 1\tr \boldsymbol y)$.
Initiating from $\gamma^{(0)} = 1$ and $\boldsymbol \delta^{(0)} = \boldsymbol 1$, and alternating the core step, GIS($\gamma^{(d)}$), and  $\gamma$-adjustment step: $$\gamma^{(d+1)} =  \gamma^{(d)} - \frac{\boldsymbol 1' \tilde{\boldsymbol \delta}^{(d)}}{\gamma^{(d)}(\mathbf{A} \boldsymbol q)\tr  (\mathbf{A} \,  diag (\,{\tilde{\boldsymbol \delta}^{(d)}}) \, \mathbf{A}\tr)^{-1} (\mathbf{A} \boldsymbol q)},$$ obtain the sequence $\tilde{\boldsymbol \delta}_{\gamma^{(d)}}$, where, for each $d \geq 0$, $\tilde{\boldsymbol \delta}_{\gamma^{(d)}}$ is the limit of GIS($\gamma^{(d)}$). 
Then, as $d \to \infty$, the sequence $\tilde{\boldsymbol \delta}_{\gamma^{(d)}}$ converges, and its limit is equal to the MLE   $\hat{\boldsymbol p}$ under the model $\mathcal{G}(\mathbf{B})$.
 \end{theorem}

\noindent The statement directly follows from Theorem \ref{ThGammaNew}. In the following example, the MLE computation using GIS($\gamma$) is illustrated in both cases (intensities and probabilities). 


\begin{example}\label{ExampleForPlots}
Consider a general log-linear model specified by the model matrix 
$$\mathbf{A} = \left( \begin{array}{cccc}
                                   1 & 0 & 3 & 2\\
                                   1 & 3 & 0 & 2\\
                                   \end{array} \right).
                                   $$
Because the row of $1$'s is not in the row span of $\mathbf{A}$, the model $\mathcal{G}(\mathbf{A})$ does not have the overall effect. Therefore, previously proposed iterative scaling algorithms, even those which allow for non 0-1 entries in model matrix,  cannot be used to find the MLE.

For  illustration purposes, two models are of interest: $\mathcal{G}(\mathbf{A})$ as a model for intensities and as a model for probabilities. Both are models without the overall effect, since $\boldsymbol 1\tr \notin rowspan(\mathbf{A})$.  Let the observed frequency distribution be $\boldsymbol y$ = $(1,2,3,4)\tr$. 

In the context of intensities,   GIS($1$), applied directly to the data $\boldsymbol y$, converges to the vector  $\tilde{\boldsymbol \lambda}$ =  $(1.8575, 2.0805, 3.0806, 3.4504)\tr$, such that, up to a precision, $\mathbf{A} \tilde{\boldsymbol \lambda} = \mathbf{A} \boldsymbol y$ and $\mathbf{D} \log \tilde{\boldsymbol \lambda} = \boldsymbol 0$. By uniqueness,  $\tilde{\boldsymbol \lambda} = \hat{\boldsymbol \lambda}$ is the MLE. Notice that the estimated total is not equal to the observed one, 
$\boldsymbol 1\tr \hat{\boldsymbol \lambda} = 10.4690 \neq 10$. Until convergence, 41 iterations were required.

In order to find the MLE in the probabilities case, the iterative scaling has to be applied to normalized data, $\boldsymbol q = \boldsymbol y/\boldsymbol 1\tr \boldsymbol y$. To initiate, choose, for example, $\gamma =1$ and apply GIS($1$).  For $\boldsymbol q = (0.1, 0.2, 0.3, 0.4)\tr$,  GIS($1$) converges to the vector $\tilde{\boldsymbol\delta}_1$ = $(0.4284, 0.2348, 0.3348, 0.1835)\tr$, for which $\mathbf{A}\tilde{\boldsymbol{\delta}}_1 =  \mathbf{A} \boldsymbol q$ and $\mathbf{D} \mbox{log } \tilde{\boldsymbol{\delta}}_1 = \boldsymbol 0$. However, because the total ${\boldsymbol{1}}\tr \tilde{\boldsymbol{\delta}}_1 = 1.1815 \neq 1$,  $ \tilde{\boldsymbol{\delta}}_1$ is not the MLE, so the value of $\gamma$ has to be adjusted. The subsequent application of GIS($\gamma$) and $\gamma$-adjustment step is performed until the total is as close to $1$ as desired. 
In this example, 10 adjustment steps were needed to achieve the convergence to the total of $1$ within four decimal places. Within each adjustment step, the convergence in $L_2$ norm in sufficient statistics was achieved in 37 iterations. At the convergence, the MLE for probabilities was found to be $\hat{\boldsymbol p} = (0.3799, 0.1960, 0.2798, 0.1443)\tr$, with $\hat{\gamma} = 0.8377$.   
In Figure \ref{Ex41prob}, the initial step ($\gamma = 1$) and the last adjustment step ($\gamma = 0.8377$) are depicted. Notice that the KL divergence can take negative values, so the proof of \cite{DarrochRatcliff} does not apply. \qed
\end{example}

\begin{figure}[htbp]
\begin{minipage}[t]{0.45\textwidth}
\centering
\includegraphics[scale=0.25]{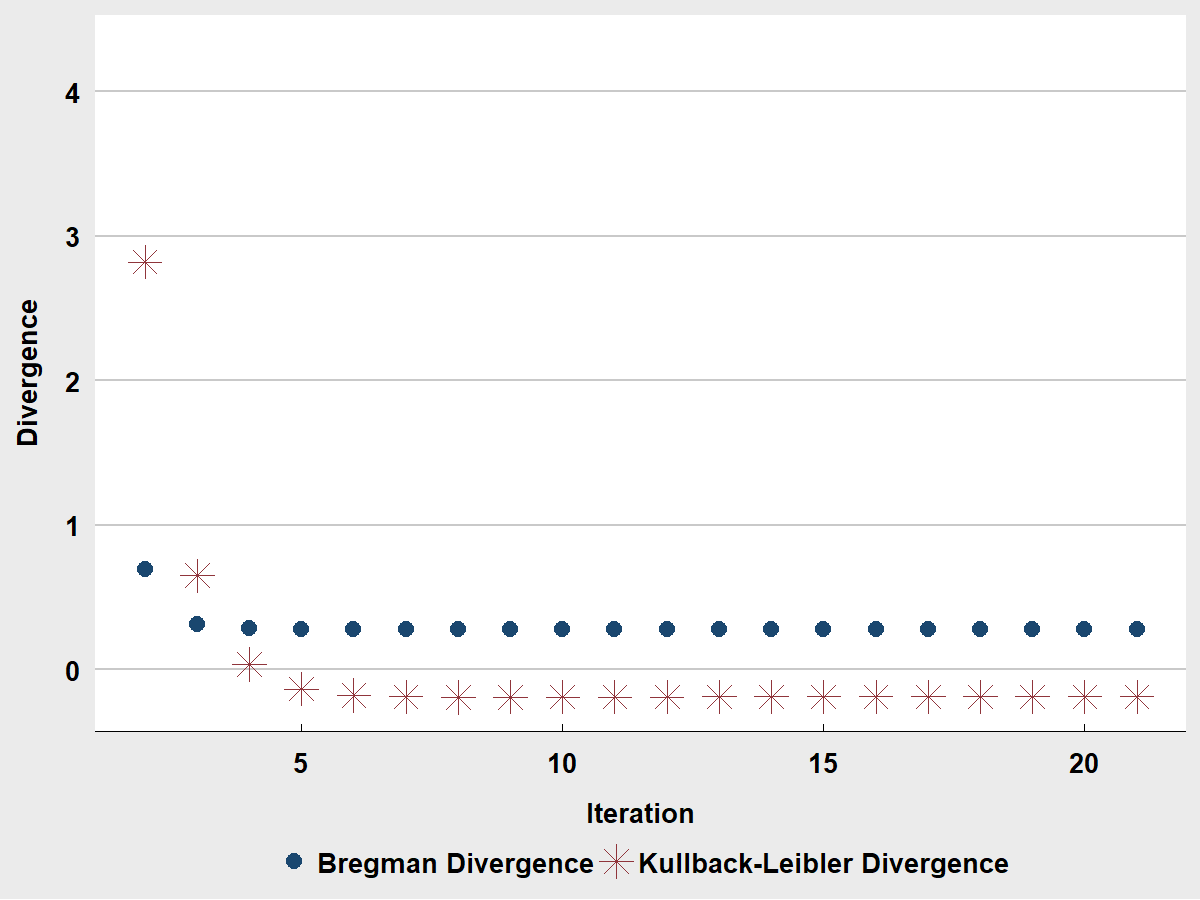}
\caption{Values of Bregman and Kullback-Leibler divergence until the convergence of GIS($\gamma$) was achieved  for the model for intensities in Example \ref{ExampleForPlots}.}
\label{Ex41intens}
\end{minipage} \hfill
\begin{minipage}[t]{0.45\textwidth}
\centering
\includegraphics[scale=0.25]{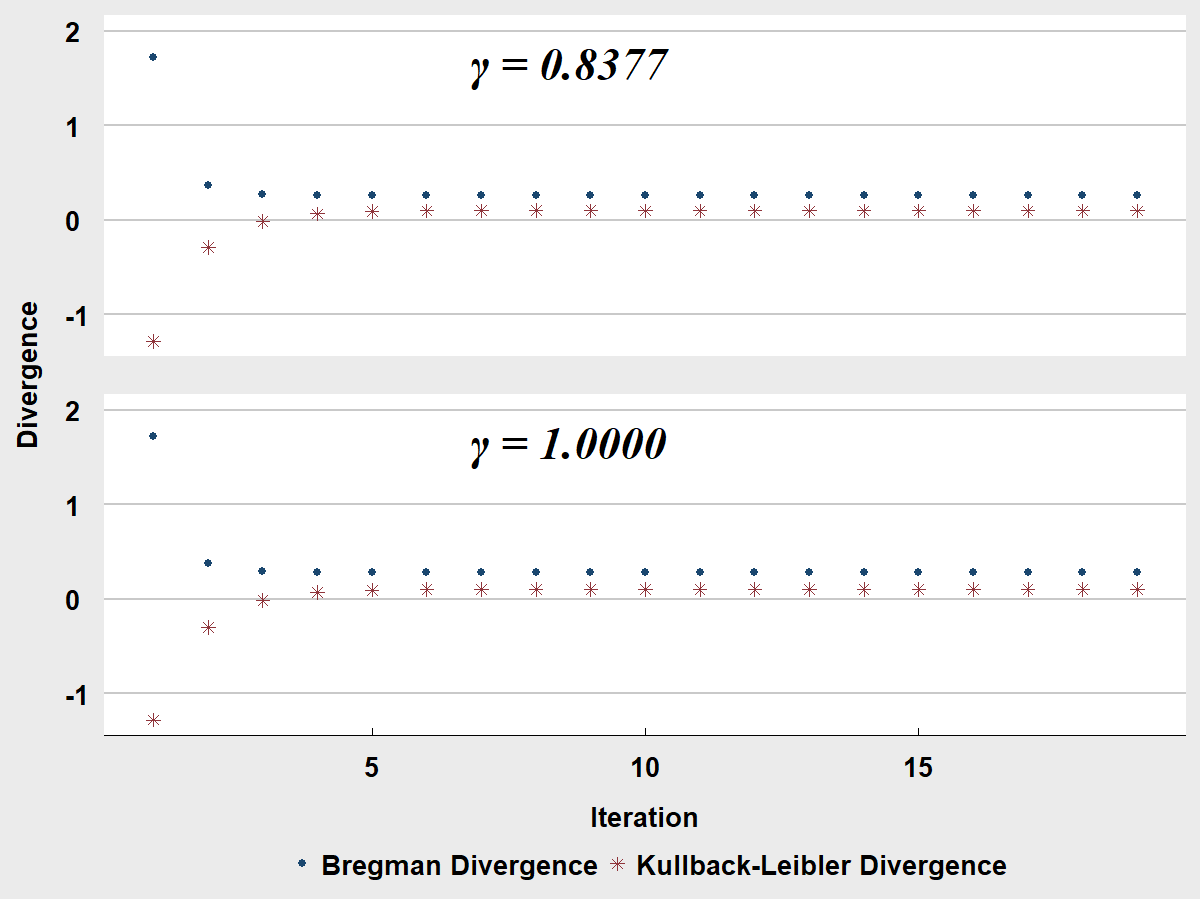}
\caption{Values of Bregman and Kullback-Leibler divergence until the convergence of GIS($\gamma$), with adjustment, was reached under the model for probabilities in Example \ref{ExampleForPlots}.}
\label{Ex41prob}
\end{minipage}
\end{figure}

The GIS ($\gamma$) algorithm  can also be applied to compute the MLE under and an affine general log-linear model which is defined as,
\begin{equation} \label{PMmatrA}
\mathcal{G}(\mathbf{A}) = \left \{ \boldsymbol \delta \in \mathcal{P}: \,\, \mbox{log } \boldsymbol \delta = \mathbf{A}'\boldsymbol \beta + \xi, \, \mbox{for some } \, \boldsymbol \beta, \boldsymbol \xi  \in \mathbb{R}^J \right\}.
\end{equation}
Its dual representation is also derived from a matrix, $\mathbf{D}$, whose rows are a basis of $Ker(\mathbf{A})$, except for that, in the affine case, the log-odds ratios are allowed to take non-zero values:
\begin{equation}\label{dualF2dA}
\mathcal{G}(\mathbf{A}) = \{ {\boldsymbol \delta} \in \mathcal{P}: \,\, \mathbf{D} \mbox{log }\boldsymbol \delta = \boldsymbol \psi\}.
\end{equation}
To illustrate, Example \ref{ExampleForPlots} is revisited next.
\begin{example}\label{ExampleMixed}
Consider the log-linear model generated by the matrix $\mathbf{A}$ equal to:
$$\mathbf{A} = \left( \begin{array}{cccc}
                                   1 & 0 & 3 & 2\\
                                   1 & 3 & 0 & 2\\
                                   \end{array} \right).
                                 $$
and with a kernel basis matrix:                         
$$\mathbf{D} = \left( \begin{array}{rrrr}
                                   2 & 0 & 0 & -1\\
                                   1 & -1 & -1 & 1\\
                                   \end{array} \right).
$$

Suppose that the data $\boldsymbol y$ come from $Mult(N, \boldsymbol p)$ and that the MLE under the model for probabilities $\mathbf{D} \log \boldsymbol p = \mathbf{D}  \log(1/12, 9/8)\tr$ is of interest.  So, one has to compute a vector, say ${\boldsymbol r} > \boldsymbol 0$, such as: 
$$\mathbf{A} \boldsymbol r = \gamma_{q} \mathbf{A} (\boldsymbol y/N),\quad \mathbf{D}\log {\boldsymbol r} = \log(1/12, 9/8)\tr, \quad  \boldsymbol 1\tr  \boldsymbol r = 1.$$
Then, by uniqueness of the MLE,  $ \boldsymbol r = \hat{ \boldsymbol p}$.

In fact, for this particular choice of the affine model, the MLE has closed form \citep[cf.][]{SturmfelsDuarteMLE2021}. Set $z_1 = y_1 + y_2+ 2y_3 + 2y_4$,   $z_2 = y_1 + 3y_3 + 2y_4$,  $z_3 = y_1 + 3y_2+ 2y_4$,  $z_4 = y_1 +2 y_2+ y_3 + 2y_4$, and take
\begin{equation}\label{mleWith27}
\boldsymbol r  =
 \left (\frac{2z_2z_3}{3z_1z_4}, \frac{4z_3^3}{27z_1 z_4^2}, \frac{4z_2^3}{27z_1^2z_4},  \frac{z_2^2 z_3^2}{27z_1^2 z_4^2}\right)\tr.
 \end{equation}  
It is straightforward to check that $\boldsymbol 1\tr  \boldsymbol r = 1$ and, because $8 r_1r_4 = 9 r_2 r_3$, $r_1^2 = 12 r_4$, that is, $\mathbf{D}\log {\boldsymbol r} = \log(1/12, 9/8)\tr$, so  $\boldsymbol r$ is a probability distribution in the model of interest.
It remains to show that $\mathbf{A} \boldsymbol r = \gamma_{y} \mathbf{A} (\boldsymbol y/N)$.  Indeed,
    \begin{align*}
   r_1 + 3r_3 + 2r_4 &= \frac{2z_2z_3}{3z_1z_4} + 3 \frac{4z_2^3}{27z_1^2z_4} + 2 \frac{z_2^2 z_3^2}{27z_1^2 z_4^2} = \left (\frac{2z_3}{3z_1z_4} +3 \frac{4z_2^2}{27z_1^2z_4}+2 \frac{z_2 z_3^2}{27z_1^2 z_4^2}\right) z_2 = \gamma_1  \frac{z_2}{N},\\
  r_1 + 3r_2 + 2r_4 &=  \frac{2z_2z_3}{3z_1z_4} + 3 \frac{4z_3^3}{27z_1z_4^2} + 2 \frac{z_2^2 z_3^2}{27z_1^2 z_4^2} = \left (\frac{2z_2}{3z_1z_4} +3 \frac{4z_3^2}{27z_1z_4^2}+2 \frac{z_2^2 z_3}{27z_1^2 z_4^2}\right) z_3 = \gamma_2  \frac{z_3}{N}, 
 \end{align*}    
where
\begin{align*}
\gamma_1 = \gamma_2 &=\frac{N}{27z_1^2 z_4^2} (32z_1^3 + 144z_1^2z_2 + 144z_1^2z_3 + 192 z_1^2z_4 + 216 z_1z_2^2 + 432 z_1z_2z_3 + 576 z_1z_2z_4 \\ 
                                               &+  216 z_1z_3^2 + 576z_1z_3z_4 + 384 z_1z_4^2 + 108 z_2^3 + 324 z_2^2 z_3 + 432 z_2^2 z_4 + 324 z_2z_3^2 \\
                                               &+ 864 z_2z_3z_4 + 576 z_2z_4^2 + 108 z_3^3 + 432z_3^2z_4 + 576 z_3z_4^2 + 256 z_4^3),
\end{align*}
implying that $\mathbf{A} \boldsymbol r = \gamma_{y} \mathbf{A} (\boldsymbol y/N)$, and thus, $ \boldsymbol r = \hat{ \boldsymbol p}$.

Let the observed data be $\boldsymbol y = (1, 2, 3, 4)\tr$.  In order to commence GIS($\gamma$) for an affine model, one can start with any vector with integer entries which satisfy the desired multiplicative structure. Here, one can, for example, use $\boldsymbol r = (6, 4, 4, 3)\tr$.

After rounding to four decimal places either the closed-form solution or the its approximation obtained from GIS($\gamma$),  $\boldsymbol p \approx  (0.6618, 0.1149, 0.1869, 0.0365)\tr$ and $\gamma \approx 0.7196$. During the numerical fitting,  the adjustment step was performed 133 times,  the convergence of each core step was reached in 53 iterations.\qed
\end{example}

The proposed framework will now be used to analyze the data obtained from a ongoing clinical study  \citep{DIAVacc}.  The participants were vaccinated  against the COVID-19 disease using an mRNA vaccine and have been monitored since then, with the primary focus on specific immune response rates to vaccination.  The study participants who did not respond to the primary vaccination, were revaccinated, \textit{boostered}, after six month, and those who did not respond to the first booster, were revaccinated one more time. Therefore, the (re-) vaccination design can be described using the tree in  Figure \ref{TreeVaccine}.  Denoted by $Y_1$, $Y_2$, and $Y_3$ are the random variables corresponding to  the responses to primary vaccination, and the first and second revaccinations, respectively. Of special interest is the data of kidney transplant recipients whose rate of positive response to the initial vaccination was the lowest among the study groups. The possibility of delayed response in these patients seems plausible, and, therefore, in this paper, we will test the hypothesis of independence of immune response in the course of repeated vaccination, which can be expressed using the model in Example \ref{ExampleVaccine}:
$$p_{1} = \theta_0^3, \,\,  p_{2} = \theta_0^2 \theta_1, \,\, p_{3} = \theta_0\theta_1, \,\, p_{4} = \theta_1,$$
where $\theta_1$ denotes the probability of positive immune response, and $\theta_0= 1-\theta_1$ is the probability of non-response.  Since the number of participants was fixed at the beginning of the study, the multinomial sampling scheme can be assumed, and, for this case, the MLE under this model has a closed form and was computed in Example \ref{ExampleVaccine}. 
Given the observed frequency distribution $\boldsymbol y = (80, 12, 44, 64)\tr$,  the MLE is equal to $\hat{\boldsymbol p}$ = $((308/428)^3$, $308^2 \cdot 120/428^3$, $308 \cdot 120/428^2$, $120/428)\tr$. In this case, Pearson statistic $X^2 = 11.85$  and deviance $G^2 = 14.65$ on two degrees of freedom, which indicates a medium-strength evidence against the independence of immune response due to repeated vaccination. The MLE computation was also performed using the GIS($\gamma$) algorithm.  Three adjustment steps with 59 iterations each were needed to achieve the convergence within four decimal places.  Upon convergence, $\hat{\boldsymbol p}\approx (0.3727, 0.1452, 0.2018, 0.2804)\tr$, and $\hat{\gamma} \approx 1.0455$, both being sufficiently close to their values derived from the closed form. 

The dual representation of this model can be obtained by setting two odds ratios to unity, namely, $p_1 p_3 p_4/p_2^2 = 1$ and $p_2 p_4 = p_3^2$ \citep{KRD11}.  For the purpose of statistical power analysis, one may wish to investigate a probability distribution at least one of whose odds ratios slightly deviate from $1$, for example, a probability distribution $\tilde{\boldsymbol p} \in \mathcal{P}$ which satisfies $\log \tilde{\boldsymbol p} = \mathbf{A}\boldsymbol \beta + \beta_0$, that is $p_1 p_3 p_4/p_2^2 = 1$ and $p_2 p_4 = \phi p_3^2$, for some $\phi > 0$. Because the investigation of whether or not there is a closed-form solution for $\tilde{\boldsymbol p}$ is not trivial,  one can simply carry out the computation using GIS($\gamma$), initiating it from a $\boldsymbol \delta$ which meets the given odds ratio constraints. For instance, after setting $\phi = 2$ one obtains $\tilde{\boldsymbol p} = (0.3847, 0.1376, 0.1501, 0.3277)$ which satisfies, up to a precision, the aforementioned  odds ratio constraints and has $\mathbf{A} \tilde{\boldsymbol p} = \tilde{\gamma} \mathbf{A} ({\boldsymbol y}/\boldsymbol 1\tr \boldsymbol y)$, where $\tilde{\gamma} = 1.0255$. This distribution can be used further to estimate the test power to reject the hypothesis of independence.

\section{Conclusion}

The paper reviews several aspects  of the maximum likelihood estimation under general log-linear models that are specified using a non-negative integer design matrix, mainly focusing on the models which do not have the genuine overall effect. In such models, the sampling scheme has a profound effect on the MLE properties, and, in particular, on how the MLE is calculated. Although the Generalized Iterative Scaling (GIS) of \cite{DarrochRatcliff} would be a natural choice for the MLE computation for non-negative design matrices with integer entries,  the original proof of convergence of this algorithm was given only for models with the overall effect. In this paper we prove that, in fact, the algorithm does converge in the no-overall-effect case, and is sufficient for the Poisson context. An extended version of GIS which can be used to compute the MLE in the multinomial case is proposed as well.  Several examples are provided to illustrate the main points and results. Computations were performed using the R Environment for statistical computing \citep{Rcite}.

\section*{Acknowledgements}

The authors are grateful to Tam\'{a}s Rudas for his insightful comments during the preparation of this manuscript. 
The data set is courtesy of Christian Hugo and Julian Stumpf, from the University Clinic Carl-Gustav Carus, Technical University Dresden. The authors are indebted to Ren\'{e} Mauer for the data processing.

\clearpage
 
\section*{Appendix}


\vspace{4mm}

\noindent \textbf{Proof of Lemma \ref{LemmaBregmanInequality}}

\vspace{3mm}

\begin{proof}
 \begin{align*}
 \mathcal{D}(\boldsymbol z, \boldsymbol \delta^{(n+1)}) &= \sum_{i} z_i \log z_i -  \sum_{i} z_i \log \delta_i^{(n +1)} - (\sum_{i} z_i  - \sum_{i} \delta_i^{(n +1)}) \\
 &\leq \sum_{i} z_i \log z_i -  \sum_{i} z_i \log \delta_i^{(n +1)} - \sum_{i} z_i  + \sum_{i =1}^I\delta_i^{(n)} +  \sum_{j = 1}^J   {A}_{j} \boldsymbol{q} - \sum_{j = 1}^J A_j \boldsymbol \delta^{(n)}  \\
 &= \sum_{i} z_i \log z_i  - \sum_{i} z_i \log  \delta_i^{(n)} \prod_{j = 1}^J \left[\frac{{A}_{j} \boldsymbol{q}}{{A}_{j} \boldsymbol{\delta}^{(n)}}\right]^{a_{ji}} \\
 &- \sum_{i} z_i  + \sum_{i =1}^I\delta_i^{(n)} +  \sum_{j = 1}^J   {A}_{j} \boldsymbol{q} - \sum_{j = 1}^J A_j \boldsymbol \delta^{(n)}\\
 &= \sum_{i} z_i \log z_i  - \sum_{i} z_i \log  \delta_i^{(n)} -   \sum_{i} z_i \sum_{j = 1}^J{a_{ji}}  \log \left[\frac{{A}_{j} \boldsymbol{q}}{{A}_{j} \boldsymbol{\delta}^{(n)}}\right] \\
 &- \sum_{i} z_i  + \sum_{i =1}^I\delta_i^{(n)} +  \sum_{j = 1}^J   {A}_{j} \boldsymbol{q} - \sum_{j = 1}^J A_j \boldsymbol \delta^{(n)}\\
 &= \sum_{i} z_i \log (z_i/\delta_i^{(n)}) - (\sum_{i} z_i  - \sum_{i =1}^I\delta_i^{(n)})-   \sum_{i} z_i \sum_{j = 1}^J{a_{ji}}  \log \left[\frac{{A}_{j} \boldsymbol{q}}{{A}_{j} \boldsymbol{\delta}^{(n)}}\right] \\
 &+  \sum_{j = 1}^J   {A}_{j} \boldsymbol{q} - \sum_{j = 1}^J A_j \boldsymbol \delta^{(n)}\\
 &= \mathcal{D}(\boldsymbol z, \boldsymbol \delta^{(n)}) - 
 \sum_{j = 1}^J   \log \left[\frac{{A}_{j} \boldsymbol{q}}{{A}_{j} \boldsymbol{\delta}^{(n)}}\right]\sum_{i} z_i{a_{ji}} + \sum_{j = 1}^J   {A}_{j} \boldsymbol{q} - \sum_{j = 1}^J A_j \boldsymbol \delta^{(n)}  \\
 & = \mathcal{D}(\boldsymbol z, \boldsymbol \delta^{(n)}) - 
 \sum_{j = 1}^J  A_j \boldsymbol z \cdot \log \left[\frac{{A}_{j} \boldsymbol{q}}{{A}_{j} \boldsymbol{\delta}^{(n)}}\right] + \sum_{j = 1}^J   {A}_{j} \boldsymbol{q} - \sum_{j = 1}^J A_j \boldsymbol \delta^{(n)} .
  \end{align*}
 
 Therefore, after taking $\boldsymbol z \equiv \boldsymbol q$,
 $$\mathcal{D}(\boldsymbol q, \boldsymbol \delta^{(n+1)}) \leq  \mathcal{D}(\boldsymbol q, \boldsymbol \delta^{(n)}) - \mathcal{D}(\mathbf{A}\boldsymbol q, \mathbf{A}\boldsymbol \delta^{(n)}).$$
 
\end{proof}

\bibliography{IPFmanuscript20221230arxiv.bib}

\begin{thebibliography}{}

\bibitem [\protect \citeauthoryear {%
Agresti%
}{%
Agresti%
}{%
{\protect \APACyear {2002}}%
}]{%
Agresti2002}
\APACinsertmetastar {%
Agresti2002}%
\begin{APACrefauthors}%
Agresti, A.%
\end{APACrefauthors}%
\unskip\
\newblock
\APACrefYear{2002}.
\newblock
\APACrefbtitle {{C}ategorical {D}ata {A}nalysis} {{C}ategorical {D}ata
  {A}nalysis}.
\newblock
\APACaddressPublisher{New York}{Wiley}.
\PrintBackRefs{\CurrentBib}

\bibitem [\protect \citeauthoryear {%
Aitchison%
\ \BBA {} Silvey%
}{%
Aitchison%
\ \BBA {} Silvey%
}{%
{\protect \APACyear {1960}}%
}]{%
AitchSilvey60}
\APACinsertmetastar {%
AitchSilvey60}%
\begin{APACrefauthors}%
Aitchison, J.%
\BCBT {}\ \BBA {} Silvey, S\BPBI D.%
\end{APACrefauthors}%
\unskip\
\newblock
\APACrefYearMonthDay{1960}{}{}.
\newblock
{\BBOQ}\APACrefatitle {Maximum-likelihood estimation procedures and associated
  tests of significance} {Maximum-likelihood estimation procedures and
  associated tests of significance}.{\BBCQ}
\newblock
\APACjournalVolNumPages{J. R. Stat. Soc. Ser.B}{22}{}{154--171}.
\PrintBackRefs{\CurrentBib}

\bibitem [\protect \citeauthoryear {%
Andersen%
}{%
Andersen%
}{%
{\protect \APACyear {1974}}%
}]{%
Andersen74}
\APACinsertmetastar {%
Andersen74}%
\begin{APACrefauthors}%
Andersen, A\BPBI H.%
\end{APACrefauthors}%
\unskip\
\newblock
\APACrefYearMonthDay{1974}{}{}.
\newblock
{\BBOQ}\APACrefatitle {Multidimensional contingency tables} {Multidimensional
  contingency tables}.{\BBCQ}
\newblock
\APACjournalVolNumPages{Scand. J. Statist.}{1}{}{115--127}.
\PrintBackRefs{\CurrentBib}

\bibitem [\protect \citeauthoryear {%
Bertsekas%
}{%
Bertsekas%
}{%
{\protect \APACyear {1999}}%
}]{%
BertsekasNLP}
\APACinsertmetastar {%
BertsekasNLP}%
\begin{APACrefauthors}%
Bertsekas, D\BPBI P.%
\end{APACrefauthors}%
\unskip\
\newblock
\APACrefYear{1999}.
\newblock
\APACrefbtitle {Nonlinear programming} {Nonlinear programming}.
\newblock
\APACaddressPublisher{}{Athena Scientific}.
\PrintBackRefs{\CurrentBib}

\bibitem [\protect \citeauthoryear {%
Birch%
}{%
Birch%
}{%
{\protect \APACyear {1963}}%
}]{%
BirchMLE}
\APACinsertmetastar {%
BirchMLE}%
\begin{APACrefauthors}%
Birch, M\BPBI W.%
\end{APACrefauthors}%
\unskip\
\newblock
\APACrefYearMonthDay{1963}{}{}.
\newblock
{\BBOQ}\APACrefatitle {Maximum likelihood in three-way contingency tables}
  {Maximum likelihood in three-way contingency tables}.{\BBCQ}
\newblock
\APACjournalVolNumPages{J. Roy. Statist. Soc. Ser.B}{25}{}{220--233}.
\PrintBackRefs{\CurrentBib}

\bibitem [\protect \citeauthoryear {%
Bishop%
, Fienberg%
\BCBL {}\ \BBA {} Holland%
}{%
Bishop%
\ \protect \BOthers {.}}{%
{\protect \APACyear {1975}}%
}]{%
BFH}
\APACinsertmetastar {%
BFH}%
\begin{APACrefauthors}%
Bishop, Y\BPBI M\BPBI M.%
, Fienberg, S\BPBI E.%
\BCBL {}\ \BBA {} Holland, P\BPBI W.%
\end{APACrefauthors}%
\unskip\
\newblock
\APACrefYear{1975}.
\newblock
\APACrefbtitle {Discrete multivariate analysis: {theory} and practice}
  {Discrete multivariate analysis: {theory} and practice}.
\newblock
\APACaddressPublisher{Cambridge, MA, USA}{MIT}.
\PrintBackRefs{\CurrentBib}

\bibitem [\protect \citeauthoryear {%
Bregman%
}{%
Bregman%
}{%
{\protect \APACyear {1967}}%
}]{%
Bregman}
\APACinsertmetastar {%
Bregman}%
\begin{APACrefauthors}%
Bregman, L\BPBI M.%
\end{APACrefauthors}%
\unskip\
\newblock
\APACrefYearMonthDay{1967}{}{}.
\newblock
{\BBOQ}\APACrefatitle {The relaxation method of finding the common point of
  convex sets and its application to the solution of problems in convex
  programming} {The relaxation method of finding the common point of convex
  sets and its application to the solution of problems in convex
  programming}.{\BBCQ}
\newblock
\APACjournalVolNumPages{{U.S.S.R. Computational Mathematics and Mathematical
  Physics}}{7}{3}{200--217}.
\PrintBackRefs{\CurrentBib}

\bibitem [\protect \citeauthoryear {%
Csiszar%
}{%
Csiszar%
}{%
{\protect \APACyear {1975}}%
}]{%
Csiszar}
\APACinsertmetastar {%
Csiszar}%
\begin{APACrefauthors}%
Csiszar, I.%
\end{APACrefauthors}%
\unskip\
\newblock
\APACrefYearMonthDay{1975}{}{}.
\newblock
{\BBOQ}\APACrefatitle {{$I$-Divergence} Geometry of probability distributions
  and minimization problems} {{$I$-Divergence} geometry of probability
  distributions and minimization problems}.{\BBCQ}
\newblock
\APACjournalVolNumPages{The Annals of Probability}{3 (1)}{}{146--158}.
\PrintBackRefs{\CurrentBib}

\bibitem [\protect \citeauthoryear {%
Darroch%
\ \BBA {} Ratcliff%
}{%
Darroch%
\ \BBA {} Ratcliff%
}{%
{\protect \APACyear {1972}}%
}]{%
DarrochRatcliff}
\APACinsertmetastar {%
DarrochRatcliff}%
\begin{APACrefauthors}%
Darroch, J\BPBI N.%
\BCBT {}\ \BBA {} Ratcliff, D.%
\end{APACrefauthors}%
\unskip\
\newblock
\APACrefYearMonthDay{1972}{}{}.
\newblock
{\BBOQ}\APACrefatitle {Generalized iterative scaling for log-linear models}
  {Generalized iterative scaling for log-linear models}.{\BBCQ}
\newblock
\APACjournalVolNumPages{The Annals of Mathematical
  Statistics}{43}{}{1470--1480}.
\PrintBackRefs{\CurrentBib}

\bibitem [\protect \citeauthoryear {%
Della~Pietra%
, Della~Pietra%
\BCBL {}\ \BBA {} Lafferty%
}{%
Della~Pietra%
\ \protect \BOthers {.}}{%
{\protect \APACyear {1997}}%
}]{%
DDL1997}
\APACinsertmetastar {%
DDL1997}%
\begin{APACrefauthors}%
Della~Pietra, S.%
, Della~Pietra, V.%
\BCBL {}\ \BBA {} Lafferty, J.%
\end{APACrefauthors}%
\unskip\
\newblock
\APACrefYearMonthDay{1997}{}{}.
\newblock
{\BBOQ}\APACrefatitle {Inducing features of random fields} {Inducing features
  of random fields}.{\BBCQ}
\newblock
\APACjournalVolNumPages{{IEEE Trans. Pattern Analysis and Machine
  Intelligence}}{19}{}{283--297}.
\PrintBackRefs{\CurrentBib}

\bibitem [\protect \citeauthoryear {%
Duarte%
, Marigliano%
\BCBL {}\ \BBA {} Sturmfels%
}{%
Duarte%
\ \protect \BOthers {.}}{%
{\protect \APACyear {2021}}%
}]{%
SturmfelsDuarteMLE2021}
\APACinsertmetastar {%
SturmfelsDuarteMLE2021}%
\begin{APACrefauthors}%
Duarte, E.%
, Marigliano, O.%
\BCBL {}\ \BBA {} Sturmfels, B.%
\end{APACrefauthors}%
\unskip\
\newblock
\APACrefYearMonthDay{2021}{}{}.
\newblock
{\BBOQ}\APACrefatitle {Discrete statistical models with rational maximum
  likelihood estimator} {Discrete statistical models with rational maximum
  likelihood estimator}.{\BBCQ}
\newblock
\APACjournalVolNumPages{Bernoulli}{27}{1}{135--154}.
\PrintBackRefs{\CurrentBib}

\bibitem [\protect \citeauthoryear {%
Evans%
\ \BBA {} Forcina%
}{%
Evans%
\ \BBA {} Forcina%
}{%
{\protect \APACyear {2013}}%
}]{%
EvansForcina11}
\APACinsertmetastar {%
EvansForcina11}%
\begin{APACrefauthors}%
Evans, R\BPBI J.%
\BCBT {}\ \BBA {} Forcina, A.%
\end{APACrefauthors}%
\unskip\
\newblock
\APACrefYearMonthDay{2013}{}{}.
\newblock
{\BBOQ}\APACrefatitle {Two algorithms for fitting constrained marginal models}
  {Two algorithms for fitting constrained marginal models}.{\BBCQ}
\newblock
\APACjournalVolNumPages{Comput. Statist. Data Anal.}{66}{}{1--7}.
\PrintBackRefs{\CurrentBib}

\bibitem [\protect \citeauthoryear {%
Forcina%
}{%
Forcina%
}{%
{\protect \APACyear {2019}}%
}]{%
Forcina2019}
\APACinsertmetastar {%
Forcina2019}%
\begin{APACrefauthors}%
Forcina, A.%
\end{APACrefauthors}%
\unskip\
\newblock
\APACrefYearMonthDay{2019}{}{}.
\newblock
{\BBOQ}\APACrefatitle {Estimation and testing of multiplicative models for
  frequency data} {Estimation and testing of multiplicative models for
  frequency data}.{\BBCQ}
\newblock
\APACjournalVolNumPages{Metrika}{}{}{1--16}.
\newblock
\APACrefnote{{doi:} 10.1007/s00184-019-00709-6}
\PrintBackRefs{\CurrentBib}

\bibitem [\protect \citeauthoryear {%
Haber%
}{%
Haber%
}{%
{\protect \APACyear {1984}}%
}]{%
Haber}
\APACinsertmetastar {%
Haber}%
\begin{APACrefauthors}%
Haber, M.%
\end{APACrefauthors}%
\unskip\
\newblock
\APACrefYearMonthDay{1984}{}{}.
\newblock
{\BBOQ}\APACrefatitle {Algorithm {AS 207:} fitting a general log-linear model}
  {Algorithm {AS 207:} fitting a general log-linear model}.{\BBCQ}
\newblock
\APACjournalVolNumPages{J. R. Stat. Soc. Ser. C}{33(3)}{}{358--362}.
\PrintBackRefs{\CurrentBib}

\bibitem [\protect \citeauthoryear {%
Haberman%
}{%
Haberman%
}{%
{\protect \APACyear {1974}}%
}]{%
Haberman}
\APACinsertmetastar {%
Haberman}%
\begin{APACrefauthors}%
Haberman, S\BPBI J.%
\end{APACrefauthors}%
\unskip\
\newblock
\APACrefYear{1974}.
\newblock
\APACrefbtitle {The analysis of frequency data} {The analysis of frequency
  data}.
\newblock
\APACaddressPublisher{}{The University of Chicago Press}.
\PrintBackRefs{\CurrentBib}

\bibitem [\protect \citeauthoryear {%
Huang%
, Hsieh%
, Chang%
\BCBL {}\ \BBA {} Lin%
}{%
Huang%
\ \protect \BOthers {.}}{%
{\protect \APACyear {2010}}%
}]{%
Huang2010}
\APACinsertmetastar {%
Huang2010}%
\begin{APACrefauthors}%
Huang, F\BPBI L.%
, Hsieh, C\BPBI J.%
, Chang, K\BPBI W.%
\BCBL {}\ \BBA {} Lin, C\BPBI J.%
\end{APACrefauthors}%
\unskip\
\newblock
\APACrefYearMonthDay{2010}{}{}.
\newblock
{\BBOQ}\APACrefatitle {Iterative scaling and coordinate descent methods for
  maximum entropy models} {Iterative scaling and coordinate descent methods for
  maximum entropy models}.{\BBCQ}
\newblock
\APACjournalVolNumPages{J. Mach. Learn. Res.}{11}{}{815--848}.
\PrintBackRefs{\CurrentBib}

\bibitem [\protect \citeauthoryear {%
Kass%
\ \BBA {} Vos%
}{%
Kass%
\ \BBA {} Vos%
}{%
{\protect \APACyear {1997}}%
}]{%
KassVos}
\APACinsertmetastar {%
KassVos}%
\begin{APACrefauthors}%
Kass, R\BPBI E.%
\BCBT {}\ \BBA {} Vos, P\BPBI W.%
\end{APACrefauthors}%
\unskip\
\newblock
\APACrefYear{1997}.
\newblock
\APACrefbtitle {Geometrical foundations of asymptotic inference} {Geometrical
  foundations of asymptotic inference}.
\newblock
\APACaddressPublisher{New York}{Wiley}.
\PrintBackRefs{\CurrentBib}

\bibitem [\protect \citeauthoryear {%
Klimova%
\ \BBA {} Rudas%
}{%
Klimova%
\ \BBA {} Rudas%
}{%
{\protect \APACyear {2015}}%
}]{%
KRipf1}
\APACinsertmetastar {%
KRipf1}%
\begin{APACrefauthors}%
Klimova, A.%
\BCBT {}\ \BBA {} Rudas, T.%
\end{APACrefauthors}%
\unskip\
\newblock
\APACrefYearMonthDay{2015}{}{}.
\newblock
{\BBOQ}\APACrefatitle {Iterative Scaling in Curved Exponential Families}
  {Iterative scaling in curved exponential families}.{\BBCQ}
\newblock
\APACjournalVolNumPages{Scand. J. Statist.}{42}{}{832--847}.
\PrintBackRefs{\CurrentBib}

\bibitem [\protect \citeauthoryear {%
Klimova%
\ \BBA {} Rudas%
}{%
Klimova%
\ \BBA {} Rudas%
}{%
{\protect \APACyear {2018}}%
}]{%
KRoveff}
\APACinsertmetastar {%
KRoveff}%
\begin{APACrefauthors}%
Klimova, A.%
\BCBT {}\ \BBA {} Rudas, T.%
\end{APACrefauthors}%
\unskip\
\newblock
\APACrefYearMonthDay{2018}{{\APACmonth{01}}}{}.
\newblock
{\BBOQ}\APACrefatitle {On the role of the overall effect in exponential
  families} {On the role of the overall effect in exponential families}.{\BBCQ}
\newblock
\APACjournalVolNumPages{Electronic Journal of Statistics}{12}{2}{}.
\newblock
\begin{APACrefDOI} \doi{10.1214/18-EJS1453} \end{APACrefDOI}
\PrintBackRefs{\CurrentBib}

\bibitem [\protect \citeauthoryear {%
Klimova%
, Rudas%
\BCBL {}\ \BBA {} Dobra%
}{%
Klimova%
\ \protect \BOthers {.}}{%
{\protect \APACyear {2012}}%
}]{%
KRD11}
\APACinsertmetastar {%
KRD11}%
\begin{APACrefauthors}%
Klimova, A.%
, Rudas, T.%
\BCBL {}\ \BBA {} Dobra, A.%
\end{APACrefauthors}%
\unskip\
\newblock
\APACrefYearMonthDay{2012}{}{}.
\newblock
{\BBOQ}\APACrefatitle {Relational models for contingency tables} {Relational
  models for contingency tables}.{\BBCQ}
\newblock
\APACjournalVolNumPages{J. Multivariate Anal.}{104}{}{159--173}.
\PrintBackRefs{\CurrentBib}

\bibitem [\protect \citeauthoryear {%
{R Core Team}%
}{%
{R Core Team}%
}{%
{\protect \APACyear {2021}}%
}]{%
Rcite}
\APACinsertmetastar {%
Rcite}%
\begin{APACrefauthors}%
{R Core Team}.%
\end{APACrefauthors}%
\unskip\
\newblock
\APACrefYearMonthDay{2021}{}{}.
\newblock
{\BBOQ}\APACrefatitle {R: A Language and Environment for Statistical Computing}
  {R: A language and environment for statistical computing}{\BBCQ}\
  [\bibcomputersoftwaremanual].
\newblock
\APACaddressPublisher{Vienna, Austria}{}.
\newblock
\begin{APACrefURL} \url{https://www.R-project.org/} \end{APACrefURL}
\PrintBackRefs{\CurrentBib}

\bibitem [\protect \citeauthoryear {%
Rudin%
}{%
Rudin%
}{%
{\protect \APACyear {1976}}%
}]{%
Rudin}
\APACinsertmetastar {%
Rudin}%
\begin{APACrefauthors}%
Rudin, W.%
\end{APACrefauthors}%
\unskip\
\newblock
\APACrefYear{1976}.
\newblock
\APACrefbtitle {Principles of mathematical analysis} {Principles of
  mathematical analysis}.
\newblock
\APACaddressPublisher{}{McGraw-Hill}.
\PrintBackRefs{\CurrentBib}

\bibitem [\protect \citeauthoryear {%
Stumpf%
\ \protect \BOthers {.}}{%
Stumpf%
\ \protect \BOthers {.}}{%
{\protect \APACyear {2021}}%
}]{%
DIAVacc}
\APACinsertmetastar {%
DIAVacc}%
\begin{APACrefauthors}%
Stumpf, J.%
, Siepmann, T.%
, Lindner, T.%
, Karger, C.%
, Schwöbel, J.%
, Anders, L.%
\BDBL {}Hugo, C.%
\end{APACrefauthors}%
\unskip\
\newblock
\APACrefYearMonthDay{2021}{}{}.
\newblock
{\BBOQ}\APACrefatitle {Humoral and cellular immunity to {SARS-CoV-2}
  vaccination in renal transplant versus dialysis patients: A prospective,
  multicenter observational study using {mRNA-1273} or {BNT162b2 mRNA} vaccine}
  {Humoral and cellular immunity to {SARS-CoV-2} vaccination in renal
  transplant versus dialysis patients: A prospective, multicenter observational
  study using {mRNA-1273} or {BNT162b2 mRNA} vaccine}.{\BBCQ}
\newblock
\APACjournalVolNumPages{The Lanzet Regional Health Europe}{9: 100178}{}{}.
\newblock
\begin{APACrefDOI} \doi{10.1016/j.lanepe.2021.100178} \end{APACrefDOI}
\PrintBackRefs{\CurrentBib}

\end{thebibliography}
\bibliographystyle{apacite}

\end{document}